\documentclass[a4paper,UKenglish]{lipics-v2018}

\newcommand{\cameraready}[1]{{#1}}%

\nolinenumbers 

\usepackage{microtype}
\usepackage{breqn}

\usepackage{tikz}
\usetikzlibrary{positioning}

\usepackage{comment}

\def\pa{\mathsf{pa}}
\def\dep{\mathrm{depth}}

\def\dist{\mbox{\rm dist}}
\def\BP{\mathrm{BP}}
\def\DFUDS{\mathrm{DFUDS}}
\def\LOUDS{\mathrm{LOUDS}}

\def\pda{\mathsf{pda}}
\def\rank{\mathsf{rank}}
\def\select{\mathsf{select}}
\def\bpselect{\mathsf{bpselect}}

\def\rmq{\mathsf{rmq}}
\def\open{\mathsf{open}}
\def\close{\mathsf{close}}
\def\CP{\mathrm{CP}}
\def\OP{\mathrm{OP}}
\def\DFT{\mathrm{DFT}}

\def\return{\mathrm{return}}

\def\ova{\overleftrightarrow}

\newtheorem{proposition}{\bf Proposition}
\newtheorem{problem}{Problem}

\newtheorem{observation}{Observation}

\title{Dualities in Tree Representations}
\titlerunning{Tree Representation Duality} 

 \author{Rayan Chikhi}{CNRS, Universit\'e de Lille, CRIStAL, Lille, France}{rayan.chikhi@univ-lille1.fr}{}{}
\author{Alexander Sch\"onhuth}{Centrum Wiskunde \& Informatica, Amsterdam, The Netherlands}{alexander.schoenhuth@cwi.nl}{}{}

\authorrunning{R. Chikhi and A. Sch\"onhuth} 

\Copyright{Rayan Chikhi and Alexander Sch\"onhuth}

\subjclass{}


\acknowledgements{\cameraready{The authors are grateful to H\'el\`ene Touzet for helpful discussions, and to CPM reviewers for insightful comments, providing Figure~\ref{fig:t:a} and the reference to Davoodi \emph{et al}~\cite{Davoodi2017}.}}

\EventEditors{Gonzalo Navarro, David Sankoff, and Binhai Zhu}
\EventNoEds{3}
\EventLongTitle{29th Annual Symposium on Combinatorial Pattern Matching (CPM 2018)}
\EventShortTitle{CPM 2018}
\EventAcronym{CPM}
\EventYear{2018}
\EventDate{July 2--4, 2018}
\EventLocation{Qingdao, China}
\EventLogo{}
\SeriesVolume{105}
\ArticleNo{18} 

\begin{document}

\maketitle

\begin{abstract}
  A characterization of the tree $T^*$ such that
  $\BP(T^*)=\ova{\DFUDS(T)}$, the reversal of $\DFUDS(T)$ is given. An
  immediate consequence is a rigorous characterization of the tree
  $\hat{T}$ such that $\BP(\hat{T})=\DFUDS(T)$.  In summary, $\BP$ and
  $\DFUDS$ are unified within an encompassing framework, which might
  have the potential to imply future simplifications with regard to
  queries in $\BP$ and/or $\DFUDS$. Immediate benefits displayed here
  are to identify so far unnoted commonalities in most recent work on
  the Range Minimum Query problem, and to provide improvements for the
  Minimum Length Interval Query problem.
\end{abstract}

\section{Motivation}
\label{ssec.motivation}

\begin{sloppypar}
Given an array $A[1,n]$ with elements from a totally ordered set, the
Range Minimum Query (RMQ) problem is to provide a data structure that
on input positions $1\le i\le j\le n$ returns
\begin{equation}
  \label{eq.rmq}
  \rmq_A(i,j) := \min\{A[k]\mid i\le k\le j\}.
\end{equation}
In \cite{Fischer2011}, Fischer and Heun presented the first data
structure that uses $2n + o(n)$ bits and answers queries in O(1) time
(in fact, without accessing A). They first construct a tree $\cameraready{T[A]}$
(the 2D-Min-Heap of $A$). Then they observe that in a certain
parenthesis representation of $\cameraready{T[A]}$ ($\DFUDS$), the following query
leads to success for computing $\rmq_A(i,j)$ \cameraready{(where $0$ and $1$ refer
to closing and opening parentheses in $\DFUDS(\cameraready{T[A]})$, respectively)}:
\begin{align}
  \label{al.1}w_1 & \leftarrow \rmq_D(\select_0(i+1),\select_0(j))\\
  \label{al.2}\text{if} & \quad \rank_0(\open(w_1)) = i \;\text{then}\; \return\; i\\
  \label{al.3}\text{else} & \quad \return\; \rank_0(w_1)
\end{align}
where $\rmq_D$ refers to performing a range minimum query on the array
$D[x]:=\rank_1(x) - \rank_0(x)$ where $x$ indexes parentheses in
$\DFUDS(\cameraready{T[A]})$, and $1$ and $0$ represent opening and closing
parentheses, respectively. $\open(w_1)$ returns the position of the
opening parenthesis matching the one closing at position $w_1$. Note
that $D[x]-D[x-1]\in \{-1,+1\}$ for all $x\in\{2,...,2N\}$, which
turns $\rmq_D$ into an easier problem ($\pm 1$-RMQ), as was shown in
\cite{Bender00}.
\end{sloppypar}

Most recently, Ferrada and Navarro suggested an alternative approach
which leads to a shorter, hence faster query
procedure~\cite{Ferrada2017}. They construct a tree $\widehat{\cameraready{T[A]}}$ that
results from a systematic while non-trivial transformation of the
edges of $\cameraready{T[A]}$ (the number of non-root nodes $N$ remains the same).
They observed that in $\BP(\widehat{\cameraready{T[A]}})$ the following simpler query computes $\rmq_A(i,j)$:
\begin{align}
  \label{eq.ferrada1} w_2 & \leftarrow \rmq_D(\select_0(i), \select_0(j))\\
  \label{eq.ferrada2}     & \quad \return\; \rank_0(w_2)
\end{align}
The \emph{major motivation of our treatment} is the
observation---which passes unnoted in both
\cite{Ferrada2017,Fischer2011}---that
\begin{equation}
  \label{eq.motivation}
  \DFUDS(\cameraready{T[A]}) = \BP(\widehat{\cameraready{T[A]}})
\end{equation}
So, the shorter query raised by Ferrada and Gonzalez
would have worked for Fischer and Heun as well. It further raises the
question whether there are principles by which to transform trees $T$
into trees $\hat{T}$ such that
\begin{equation}
  \label{eq.hatduality}
  \DFUDS(T) = \BP(\hat{T})
\end{equation}
and, if so, what these principles look like. Here, we thoroughly
investigate related questions so as to obtain conclusive insight.  We
will show that the respective trees and their possible representations
can be juxtaposed in terms of \emph{a new duality for tree
  representations}.  In doing so, we will obtain a proof
for \eqref{eq.motivation} as an easy corollary (to consolidate our
findings, we also give a direct proof that \cite{Ferrada2017}'s query
also would have worked for \cite{Fischer2011} in Appendix
\ref{app.A}). 
In summary,  our treatment puts $\BP$ and
$\DFUDS$ into a unifying context.\\

\subsection{Related Work}
\label{ssec.related}

\noindent \textbf{RMQ's.} The RMQ problem has originally
been anchored in the study of Cartesian trees \cite{Vuillemin80},
because it is related to computing the least common ancestor (LCA) of two
nodes in a Cartesian tree derived from $A$ \cite{Gabow84}, further
complemented by the realization that any LCA computation can be cast
as an $\pm 1$-RMQ problem \cite{Berkman93} for which subsequently
further improvements were raised \cite{Munro2001,Sadakane07}. Fischer
and Heun finally established the first structure that requires
$2n+o(n)$ space and $O(1)$ time (without accessing $A$)
\cite{Fischer2011}, establishing an anchor point for many
related topics (e.g.~\cite{Navarro2013,Navarro2014}), which justified
to strive for further improvements \cite{Ferrada2017,Grossi2013}.

\textbf{Isomorphisms.}  For their latest (and likely conclusive)
improvements, \cite{Ferrada2017} made use of an isomorphism between
binary and general ordinary trees, presented in \cite{Munro2001}, and
successfully experiment with certain variations on the ground theme of
this isomorphism, to finally obtain the above-mentioned
$\widehat{\cameraready{T[A]}}$. Here, we provide an explicit treatment of these
trees, which \cite{Ferrada2017} are implicitly making use of.  From
this point of view, we provide a rigorous re-interpretation of the
treatments \cite{Ferrada2017,Fischer2011} and the links drawn with
\cite{Munro2001} therein. \cameraready{Finally, note that \cite{Davoodi2017}
  further expands on \cite{Munro2001}.}

\textbf{BP and DFUDS.} The $\BP$
representation was first presented in \cite{Jacobson1989} and
developed further in many ways (e.g.~\cite{Munro2001}). Since neither the $\BP$ nor the $\LOUDS$
\cite{Delpratt2006,Jacobson1989} representations allow for a few basic
operations relating to children and subtrees, the $\DFUDS$
representation was presented as an improvement in this regard
\cite{Benoit2005,Jansson2007}. 
\cameraready{A tree-unifying approach different to ours was proposed by Farzan \emph{et al}~\cite{Farzan2009}. 
  \cite{Davoodi2017} observes relationships between $\BP$ and $\DFUDS$ and proves them
  via the (above-mentioned) isomorphism by \cite{Munro2001}. Since our treatment avoids binary trees altogether,
it establishes a more direct approach to identifying dualities
between ordinal trees than \cite{Davoodi2017}.}

\subsection{Notation}
\label{ssec.notation}

\textbf{Trees.} Throughout, we consider rooted, ordered trees
$T=(V,E)$ (with nodes $V=V[T]$ and (directed) edges $E=E[T]$) with
root $r$. \cameraready{For the sake of notational convenience (following standard
abuse of tree notation), we will write $v\in T$ instead of $v\in V[T]$
and $T_1\subset T_2$ for $V[T_1]\subset V[T_2]$; note that induced
subgraphs do not play a relevant role in this treatment.}  By
definition of ordered trees, \emph{siblings}, that is nodes sharing
their parent node are ordered, implying the notions of left, right,
immediate right, immediate left siblings.  By $rmc_T(v)$, we denote
the rightmost child of a node $v$ in $T$ if it exists (if $T$ is
understood, we write $rmc(v)$).  Similarly, we denote by $ils_T(v)$
(or $ils(v)$ if $T$ is understood) the immediate left sibling of $v$
in $T$ if it exists. For two siblings, $u<v$ means that $u$ is left of
$v$. As usual, the partial order on siblings can be extended to a full
order, ordering all $v\in T$, by depth-first-traversal (or
breadth-first-traversal) logic, for example; here, by default, we
write $u<_Tv$ (or $u<v$ if $T$ is understood) if $u$ comes before $v$
in the depth-first traversal of $T$. We write $u=\pa(v)$ indicating
that $u$ is the parent of $v$, that is $(u,v)$ is a directed edge in
$T$.

\textbf{Parenthesis Based Tree Representations.} In the
following, we will deal with parenthesis based representations for
trees, which are vectors of opening parentheses '(' and closing
parentheses ')'. The number of opening parentheses will match
the number of closing parentheses, thereby for a tree $T$, each node $v\in T$
will be represented by a pair of opening and closing parentheses, for which
we write $\OP(v)$ and $\CP(v)$, respectively. 

The \emph{Balanced Parenthesis (BP)} representation $\BP(T)$
(e.g.~\cite{Jacobson1989,Munro2001}) is built by traversing $T$ in
depth-first order, writing an opening parenthesis when reaching a node
for the first time, and writing a closing parenthesis when reaching a
node for the second time. By depth-first order logic, this yields a
balanced representation, meaning that the number of opening matches
the number of closing parentheses (see Figure
  \ref{fig:theonlyfig}). By default, a node is identified with its opening
parenthesis $\OP(v)$.

The \emph{Depth-First Unary Degree Sequence (DFUDS)} representation
$\DFUDS(T)$ \cite{Benoit2005} is again obtained by traversing $T$ in
depth-first order, but, when reaching a node with $d$ children for the
first time, writing $d$ opening parentheses and one closing
parenthesis (and writing no parentheses when reaching it for the
second time). This sequence of parentheses becomes balanced when
appending an opening parenthesis at the beginning. It is further
convenient to identify a node with the parenthesis preceding the block
of opening parentheses that represent its children\footnote{Literature
  references are ambiguous about the exact choice of parenthesis.
  None of the alternative choices, like the first opening parenthesis
  or the closing parenthesis following the block of opening
  parentheses, would lead to any real complications also in our
  treatment.}, which for all non-root nodes is a closing
parenthesis. In other words, in $\DFUDS$, the $i$-th closing
parenthesis reflects the $i$-th non-root node in DFT order. Note that,
according to this definition, when matching opening parentheses with
closing parentheses in a balanced manner, the opening parentheses in
one block refer to the children of the closing parenthesis preceding
the block from right to left.

\textbf{Rank/Select/Open/Close.} In the following, we will treat
parenthesis vectors as bitvectors, where opening and closing
parentheses are identified with $1$ and $0$. Let $B\in\{0,1\}^n$ be a
bitvector and $x\in\{1,...,n\}$ (for enhanced exposition, running
indices run from $1$ to $n$). Then $\rank_{B,0}(x),\rank_{B,1}(x)$ are
defined to be the number of $0$'s or $1$' in $B$ up to (and including)
$B[x]$. Further, $\select_{B,0}(i),\select_{B,1}(i)$ are defined to be
the position of the $i$-th $0$ or $1$ in $B$ (if this exists). We omit
the subscript $B$ and write
$\rank_0(x),\rank_1(x),\select_0(i),\select_1(i)$ if the choice of $B$
is evident. As a relevant example (see \eqref{eq.ferrada1}), for
$\DFUDS(T)$ and $v\in T$, we have $\CP(v) = \select_0(i)$ if and only
if $\DFT(v)=i+1$, that is $v$ is the $i+1$-th node in depth-first
traversal order, also counting the root. We further write $\open(x)$
and $\close(x)$ to identify the matching partner in a (balanced
parenthesis) bitvector, that is $\open(x)$ for a position $x$ in $B$
with $B[x]=0$ is the position of the $1$ matching $x$ and vice versa
for $\close(x)$.

\subsection{Outline of Sections}
\label{ssec.outline}

We will start with the definition of a dual tree $T^*$ of $T$ in
section \ref{sec.definition}; according to this definition, $T^*$ is a
directed graph, so we still have to prove that $T^*$ is a tree, which
we will do immediately afterwards. We proceed by proving $(T^*)^*=T$,
arguably necessary for a well-defined duality. In section
\ref{ssec.treereversal}, we then show how to decompose our duality
into subdualities by introducing the definition of a reversed tree
$\ova{T}$. We conclude by providing the definition of $\hat{T}$ as the
reversed dual tree; without being able to provide a proof at this
point, note that $\hat{T}$ will turn out to be the tree from
\eqref{eq.hatduality}.

In section \ref{sec.pda}, we provide the definition of a primal-dual
ancestor, which is crucial for re-interpreting RMQ's in terms of the
notions of duality provided here. Upon having proven the unique
existence of the primal-dual ancestor in theorem \ref{t.pda}, we
re-interpret RMQ's, and beyond that not only re-interpret, but also
improve on running minimal length interval queries (MLIQ's) both in
terms of space requirements and query counts.

We will finally prove our main theorem in
section \ref{sec.bpdfuds}.

\begin{theorem}
  \label{thm:bpdfuds}
  Let $T$ be a tree and let the reversal $\ova{B}$ of a bitvector $B$
  be defined by $\ova{B}[x]:= \cameraready{1-} B[n-x+1]$, \cameraready{$\forall x \in \{1,\ldots,n\}$}. Then
  \begin{equation}
    \BP(T)=\ova{\DFUDS(T^*)}.
  \end{equation}
\end{theorem}
    
Returning to \cite{Ferrada2017}, we will finally demonstrate that
\eqref{eq.motivation}, our motivating insight, indeed holds.

\section{Tree Duality: Definition}
\label{sec.definition}

\begin{definition}[Dual tree]
  \label{def.dualtree}
  Let $T$ be a tree. The dual tree $T^*$ of $T$ is a directed graph
  that has the same vertices as $T$. Edges and order (among nodes
  sharing a parent) are given by the following rules, where we write
  $\pa^*(v)$ for the parent of $v$ in $T^*$:
\begin{itemize}
\item \textbf{Rule 1a}: The root $r$ of $T$ is also the root of $T^*$,
  that is $r$ has no parent also in $T^*$.
\item \textbf{Rule 1b}: If $v=rmc_T(r)$ then also $v=rmc_{T^*}(r)$,
  implying in particular that $\pa^*(v)=r$.
\item \textbf{Rule 2}: If $v=rmc_T(u)$ with $u\ne r$, then
  $v=ils_{T^*}(u)$, implying that $\pa^*(v)=\pa^*(u)$.
\item \textbf{Rule 3}: If $v=ils_T(u)$, then $v=rmc_{T^*}(u)$,
  implying that $\pa^*(v)=u$.
\end{itemize}
\end{definition}

\begin{remark}
  Rules 1a, 1b, 2 and 3 immediately imply that $T^*$ is a directed
  graph where each node other than $r$ has one parent. Note that the
  existence of a parent due to Rule 2 is guaranteed by induction on
  the depth of a node in $T$, where Rule 1b makes the start.
\end{remark}

\begin{figure}
\centering\includegraphics[width=.9\textwidth]{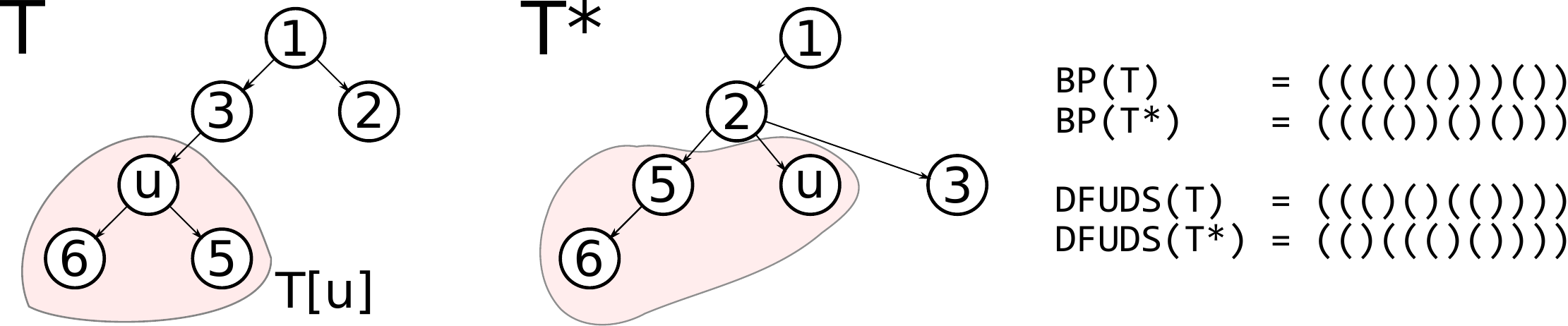}
\caption{\label{fig:theonlyfig} A tree and its dual, along with the BP and DFUDS representations. A subtree $T[u]$ is also highlighted, along with the corresponding nodes in the dual.}
\end{figure}

\begin{remark}
  It is similarly immediate to observe that there is a well-defined
  order among nodes that share a parent. It suffices to notice that in
  $T^*$ each node either is a rightmost child (Rules 1b, 3), or it is
  the (unique) immediate left sibling of another node (Rule 2).
\end{remark}

All nodes but $r$ have exactly one (incoming) edge, which implies
$|E|=|V|-1$. To conclude that $T^*$ is a tree, it remains to show
that $T^*$ contains no cycles, which we immediately do:

\begin{theorem}
  \label{t.dualtree}
  $T^*$ is a well-defined, rooted, ordered tree.
\end{theorem}

We do this by explicitly specifying the parents of nodes in $T^*$, by
making use of the depth-first traversal order $<$ in $T$. For this,
let $T[v]$ be the \emph{subtree} of $T$ that hangs off (and includes)
$v\in T$, i.e. $T[v]$ contains $v$ and all its descendants in $T$. Let further
\begin{equation*}
  R[v] := \{u\in T\setminus T[v] \mid v<u\}
\end{equation*}
be all nodes ``right of'' $v$ according to depth-first traversal
order. For two nodes $u,v$ where $u$
is an ancestor of $v$, we immediately note that
\begin{equation}
  \label{eq.subrelations}
  T[v] \subset T[u], \quad R[u] \subset R[v]\quad \text{ and } \quad R[v]\subset T[u]\;\cup\;R[u]
\end{equation}

\noindent For a node $v\in T\setminus \{r\}$, we then obtain the following lemma:

\begin{lemma}
  \label{l.dualparent}
  \begin{equation*}
    \pa^*(v) = 
    \begin{cases}
      \min R[v] & R[v]\ne\emptyset\\
      r         & R[v] = \emptyset
    \end{cases}
  \end{equation*}
\end{lemma}

\noindent We refer to 
Appendix \ref{app.applemma1} 
for the proof of
Lemma \ref{l.dualparent}. Using Lemma \ref{l.dualparent}, a proof of
theorem \ref{t.dualtree} can be immediately given:\\

\noindent \textbf{Proof of Theorem \ref{t.dualtree}.}
Lemma \ref{l.dualparent} implies that $v<_T\pa^*(v)$ for all
$v\in T\setminus \{r\}$.  Therefore, $T^*$ can contain no cycles and we
obtain that $T^*$ is a tree as a corollary. Furthermore, lemma
\ref{l.dualparent} reveals that $T^*$ is unique. \qed \\

\noindent See again 
Appendix \ref{app.applemma1} 
for immediate
corollaries 
\ref{c.suborder} and \ref{c.dualsubtree}  
which point out
how parents and subtrees in $T^*$ relate with one another.

\begin{remark}
  \label{rem.subdualities}
  An intuitive guideline for describing $T^*$ in comparison to $T$ is
  that parent- and siblinghood, as well as left and right are exchanged.
  In other words (and as will become clearer explicitly later) the
  duality describing $T^*$ can be decomposed into two subdualities, one
  of which turns parents into siblings and vice versa, and the other one
  of which exchanges left and right.
\end{remark}

This remark had left us with some choices for
characterizing tree duality. Our choice is motivated by
\cite{Fischer2011}, arguably a cornerstone in RMQ theory development.
To understand this, let $A=A[1,n]$ be the array, on which RMQ's are to
be run, and let $\ova{A}$ be its reversal, given by
$\ova{A}[i] = A[n-i+1]$. Let $\cameraready{T[A]}$ be the 2D-Min-Heap constructed
from $A$, as described in \cite{Fischer2011} (a definition is provided
in 
Appendix~\ref{app.dualminheap}), 

to which RMQ's refer (see
\eqref{al.1},\eqref{al.2},\eqref{al.3}). An immediate question to ask
is what RMQ's would look like when performing RMQ's on $\ova{A}$
instead of $A$. Here is the answer.
  
\begin{theorem}
  \label{t.dualminheap}
  Let $A[1,N]$ be an array and let
  $\ova{A}:=[A[N],...,A[1]]$ its reversal. Then
  \begin{equation}
    (\cameraready{T[A]})^* = T[\ova{A}]
  \end{equation}
\end{theorem}

\noindent \cameraready{An illustration of the Theorem is provided in Figure~\ref{fig:t:a}.} See 
Appendix \ref{app.dualminheap}  
for a more detailed
treatment of this motivating example, including proofs. Thanks to
theorem \ref{t.dualminheap}, the definition of $T^*$ can arguably be
considered a most natural choice, at least when relating tree duality with
RMQ's. \\

\begin{figure}

\begin{tikzpicture}[scale=0.9, every node/.style={scale=0.9}]

\def\y{0.35}

\foreach \x[count=\n] in {$-\infty$,2,7,8,1,6,4,3,5}{
	\node[anchor=base east,blue] (m\n) at (\n*\y,+0.4) {\small\n};
	\node[anchor=base east] (n\n) at (\n*\y,0) {\x};
	\draw ((\n*\y -0.05,-0.1) to (\n*\y -0.05,+0.3);
}

\draw (-0.4,-0.1) to (3.105,-0.1);
\draw (-0.4,0.3) to (3.105,0.3);

\draw (-0.4,-0.11) to (-0.4,0.31);

\node[anchor = base east] at (-0.8,0) {$A$};

\draw (m9.north) to[in=60 ,out=120] (m8.north);
\draw (m8.north) to[in=60 ,out=120] (m5.north);
\draw (m7.north) to[in=60 ,out=120] (m5.north);
\draw (m6.north) to[in=60 ,out=120] (m5.north);
\draw (m5.north) to[in=60 ,out=120] (m1.north);
\draw (m4.north) to[in=60 ,out=120] (m3.north);
\draw (m3.north) to[in=60 ,out=120] (m2.north);
\draw (m2.north) to[in=60 ,out=120] (m1.north);

\node[draw,circle,text=blue,inner sep=1pt] (l1) at (1,-0.5) {1};
\node[draw,circle,text=blue,inner sep=1pt,below left = 0.5cm and 0.8cm of l1] (l2)  {2};
\node[draw,circle,text=blue,inner sep=1pt,below = 0.38cm of l2] (l3)  {3};
\node[draw,circle,text=blue,inner sep=1pt,below = 0.38cm of l3] (l4)  {4};
\node[draw,circle,text=blue,inner sep=1pt,below right = 0.5cm and 0.8cm of l1] (l5)  {5};
\node[draw,circle,text=blue,inner sep=1pt,below left = 0.5cm and 0.4cm of l5] (l6)  {6};
\node[draw,circle,text=blue,inner sep=1pt,below  = 0.38cm of l5] (l7) {7};
\node[draw,circle,text=blue,inner sep=1pt,below right = 0.5cm and 0.4cm of l5] (l8)  {8};
\node[draw,circle,text=blue,inner sep=1pt,below  = 0.38cm of l8] (l9)  {9};

\draw (l1) to (l2);
\draw (l2) to (l3);
\draw (l3) to (l4);
\draw (l1) to (l5);
\draw (l5) to (l6);
\draw (l5) to (l7);
\draw (l5) to (l8);
\draw (l8) to (l9);

\node[anchor = base] at (1.5,-3.7) {$\cameraready{T[A]}$};

\foreach \x[count=\n] in {$-\infty$,5,3,4,6,1,8,7,2}{
	\node[anchor=base east,red] (pm\n) at (\n*\y+10,+0.4) {\small\n};
	\node[anchor=base east] (pn\n) at (\n*\y+10,0) {\x};
	\draw ((\n*\y -0.05+10,-0.1) to (\n*\y -0.05+10,+0.3);
}
\draw (-0.4+10,-0.1) to (3.105+10,-0.1);
\draw (-0.4+10,0.3) to (3.105+10,0.3);

\draw (-0.4+10,-0.11) to (-0.4+10,0.31);

\node[anchor = base west] at (13.6,0) {$\ova{A}$};

\draw (pm9.north) to[in=60 ,out=120] (pm6.north);
\draw (pm8.north) to[in=60 ,out=120] (pm6.north);
\draw (pm7.north) to[in=60 ,out=120] (pm6.north);
\draw (pm6.north) to[in=60 ,out=120] (pm1.north);
\draw (pm5.north) to[in=60 ,out=120] (pm4.north);
\draw (pm4.north) to[in=60 ,out=120] (pm3.north);
\draw (pm3.north) to[in=60 ,out=120] (pm1.north);
\draw (pm2.north) to[in=60 ,out=120] (pm1.north);

\node[draw,circle,text=red,inner sep=1pt] (l1) at (11,-0.5) {1};

\node[draw,circle,text=red,inner sep=1pt,below left = 0.5cm and 0.9cm of l1] (l2)  {2};
\node[draw,circle,text=red,inner sep=1pt,below = 0.38cm of l1] (l3)  {3};
\node[draw,circle,text=red,inner sep=1pt,below = 0.38cm of l3] (l4)  {4};
\node[draw,circle,text=red,inner sep=1pt,below = 0.38cm of l4] (l5)  {5};
\node[draw,circle,text=red,inner sep=1pt,below right = 0.5cm and 1.1cm of l1] (l6) {6};
\node[draw,circle,text=red,inner sep=1pt,below left = 0.5cm and 0.4cm of l6] (l7)  {7};
\node[draw,circle,text=red,inner sep=1pt,below = 0.38cm of l6] (l8)  {8};
\node[draw,circle,text=red,inner sep=1pt,below right = 0.5cm and 0.4cm of l6] (l9)  {9};

\draw (l1) to (l2);
\draw (l1) to (l3);
\draw (l3) to (l4);
\draw (l4) to (l5);
\draw (l1) to (l6);
\draw (l6) to (l7);
\draw (l6) to (l8);
\draw (l6) to (l9);

\node[anchor = base] at (11.5,-3.7) {$T[\ova{A}]$};

\node[draw,circle,text=blue,inner sep=1pt] (l1) at (6,-0.5) {1};

\node[draw,circle,text=blue,inner sep=1pt,below left = 0.5cm and 0.9cm of l1] (l2)  {9};
\node[draw,circle,text=blue,inner sep=1pt,below = 0.38cm of l1] (l3)  {8};
\node[draw,circle,text=blue,inner sep=1pt,below = 0.38cm of l3] (l4)  {7};
\node[draw,circle,text=blue,inner sep=1pt,below = 0.38cm of l4] (l5)  {6};
\node[draw,circle,text=blue,inner sep=1pt,below right = 0.5cm and 1.1cm of l1] (l6) {5};
\node[draw,circle,text=blue,inner sep=1pt,below left = 0.5cm and 0.4cm of l6] (l7)  {4};
\node[draw,circle,text=blue,inner sep=1pt,below = 0.38cm of l6] (l8)  {3};
\node[draw,circle,text=blue,inner sep=1pt,below right = 0.5cm and 0.4cm of l6] (l9)  {2};

\draw (l1) to (l2);
\draw (l1) to (l3);
\draw (l3) to (l4);
\draw (l4) to (l5);
\draw (l1) to (l6);
\draw (l6) to (l7);
\draw (l6) to (l8);
\draw (l6) to (l9);

\node[anchor = base] at (6.5,-3.7) {$(\cameraready{T[A]})^{*}$};

\end{tikzpicture}
\caption{(left) An array $A$ along with the 2D-Min-Heap $\cameraready{T[A]}$. Arcs above array indices indicate tree paths. (middle) The dual tree $(\cameraready{T[A]})^*$. (right) The reversed array $\ova{A}$ along with the 2D-Min-Heap $T[\ova{A}]$. \label{fig:t:a}}
\end{figure}

\noindent Before proceeding with results on succinct tree
representations, we provide the following intuitive lemma about the
depth-first traversal order of $T^*$ as a rooted, ordered tree. This
lemma, in combination with lemma \ref{l.dualparent}, supports the
(intended) intuition that in $T^*$ up and down, as well as left and
right, are exchanged, properties that are characteristic for rooted,
ordered tree duality.  It also provides motivation beyond theorem
\ref{t.dualminheap} in the Introduction why $T^*$ is the possibly
canonical choice of the dual of a tree.

Therefore, let $<^*$ denote the depth-first traversal order in $T^*$
(well-defined by theorem \ref{t.dualtree}) while $<$
denotes the depth-first traversal order in (the primal tree) $T$.

\begin{lemma}
  \label{l.dualorder}
  Let $u,v\in T\setminus \{r\}$. Then
  \begin{equation*}
    u <^*v \quad \text{if and only if} \quad v<u
  \end{equation*}
\end{lemma}

\noindent The proof of lemma \ref{l.dualorder} makes use of the
following technical lemmata \ref{l.dualleft} and
\ref{l.siblingsubtree}, which are of use also elsewhere.  We therefore
state these technical lemmata here. The proofs for all lemmata
\ref{l.dualorder}, \ref{l.dualleft} and \ref{l.siblingsubtree} can
finally be found in 
Appendix \ref{app.lemma3}. 

\begin{lemma}
  \label{l.dualleft}
  Let $w:=\pa^*(v)$ and $v_2\in T[v]\setminus v$ such that
  $\pa^*(v_2)=w$. Then $v_2<^*v$.
\end{lemma}

\begin{lemma}
  \label{l.siblingsubtree}
  Let $v_1$ be a sibling left of $u_1$ in $T$. Then
  $T[v_1] \subset T^*[u_1]$.
\end{lemma}

\noindent With lemma \ref{l.dualorder} proven, we can conclude with
proving a main theorem of this treatment. It states that the dual of
the dual is the primal tree, arguably a key property for a sensibly
defined duality. Despite all lemmata raised so far, the proof still
entails a few technically more demanding arguments. 

\begin{theorem}
  \label{t.dualdual}
  $(T^*)^* = T$
\end{theorem}

\begin{proof}
  It suffices to show that $\pa^{**}(v) = \pa(v)$, since lemma
  \ref{l.dualorder} establishes that the order in $(T^*)^*$ agrees with
  that of $T$. Let $u=\pa(v)$.  In 
  Appendix~\ref{app.dualofdual}, 
  we
  provide a (heavily technical) proof that
  \begin{equation*}
    u =
    \begin{cases}
      \min_{<^*} R^*[v] & R^*[v] \ne \emptyset\\
      r                 & R^*[v] = \emptyset
    \end{cases}
  \end{equation*}
  which completes the proof by applying lemma \ref{l.dualparent}.\\
\end{proof}


\subsection{Tree Reversal}
\label{ssec.treereversal}

We bring in another, simpler notion of tree duality, namely that of
reversing trees. We will further elucidate what the trees are like
when combining tree reversal with the tree duality ($T^*$) raised
earlier.

\begin{definition}[Reversed tree]
  Let $T$ be a tree. The reversed tree $\ova{T}$ of $T$ is the tree
  resulting from reversing the order among the children of each node. 
\end{definition}
\begin{proposition}
  \label{p.reversed}
  Let $\ova{T}$ be the reversed tree of $T$ and $\ova{T^*}$ be the
  reversed dual of $T$.  We define $irs$ (immediate right sibling) and $lmc$ (left-most child) similarly as in Section~\ref{ssec.notation}.
  \begin{itemize}
  \item[]\,$(a)$ The root $r$ of $T$ is also the root of $\ova{T}$.
  \item[]\,$(b)$ Let $u=\pa_T(v)$. Then also $u=\pa_{\ova{T}}(v)$.
  \item[]\,$(c)$ Let $u=ils_T(v)$. Then $u=irs_{\ova{T}}(v)$.
      \item[]\,$(d)$ The root $r$ of $T$ is also the root of
    $\ova{T^*}$.
  \item[]\,$(e)$ If $v=lmc_T(r)$ then also $v=lmc_{\ova{T^*}}(r)$,
    implying in particular that $\pa_{\ova{T^*}}(v)=r$.
  \item[]\,$(f)$ If $v=lmc_T(u)$ with $u\ne r$, so
    $v=ils_{\ova{T^*}}(u)$, implying that
    $\pa_{\ova{T^*}}(v)=\pa_{\ova{T^*}}(u)$.
  \item[]\,$(g)$ If $v=ils_T(u)$, then $v=lmc_{\ova{T^*}}(u)$,
    implying that $\pa_{\ova{T^*}}(v)=u$.
  \item[]\,$(h)$ $\ova{T^*} = \ova{T}^*$, that is the reversed dual
    tree of $T$ is the dual of the reversed tree of $T$.
  \end{itemize}
\end{proposition}

All of those are, in comparison with statements referring to the definition
of the dual tree, rather obvious observations. See 
Appendix \ref{app.reversed} 
for the proof.\\

\noindent Since $\ova{T}^*$ plays a particular role in the context of
our introductory motivation, we give it a particular name: $\hat{T}$.

\begin{definition}[Reversed dual tree]
  \label{d.reversedual}
  Let $T$ be a tree. The tree $\hat{T}:=\ova{T}^*$ of $T$ is
  the dual of the reversed (or the reversed dual) tree of $T$.
\end{definition}

\noindent Based on proposition \ref{p.reversed}, we realize that
$\hat{T}$ can be described as turning leftmost children into immediate
left siblings.

\begin{remark}
  \label{rem.ferrada}
  Following the arguments provided in \cite{Ferrada2017}, it
  becomes evident that the tree $T$ in use there, on which $\BP(T)$ is
  constructed, turns indeed out to be $\widehat{\cameraready{T[A]}}=\ova{\cameraready{T[A]}^*}$.
\end{remark}

\section{The Primal-Dual Ancestor}
\label{sec.pda}


The following theorem points out that pairs of nodes have a unique
\emph{primal-dual ancestor}. We will further point out properties
of that node.

\begin{theorem}
  \label{t.pda}
  Let $v_1, v_2\in T\setminus \{r\}$ be two nodes where $v_1\le v_2$.
  Then there is a unique node $v\in T\setminus \{r\}$ such that
  $v_1\in T^*[v]$ and $v_2\in T[v]$.
\end{theorem}

We henceforth refer to this unique node as \emph{primal-dual ancestor}
of $v_1$ and $v_2$, written $\pda(v_1,v_2)$.

\begin{proof}
  Let
  \begin{equation}
    \label{eq.pda}
    v := \max_{<_T}\{v_1\le x\le v_2\mid v_1\in T^*[x]\}
  \end{equation}
  be, relative to depth-first traversal order in $T$, the largest
  ancestor of $v_1$ in $T^*$ that precedes $v_2$. We claim that $v$
  is the unique primal-dual ancestor of $v_1$ and $v_2$.

  By definition, we immediately obtain that $v_1\in T^*[v]$. To prove
  $v_2\in T[v]$, consider $\pa^*(v)$, for which, by choice of $v$, we
  have that $v_2<\pa^*(v)$. By lemma \ref{l.dualparent}, however,
  $\pa^*(v)$ is the first node in $R[v]$, relative to depth-first
  traversal order in $T$.  Hence, for any $y$ such that
  $v\le y<\pa^*(v)$, which includes $v_2$, it holds that $y\in T[v]$.

  It remains to show that $v$ is the only possible primal-dual ancestor. By definition of the
  primal-dual ancestor, $v$ must be an ancestor of $v_1$ in $T^*$.
  
  \emph{First}, consider an ancestor $y$ of $v_1$ in $T^*$ such that $y<v$.
  By choice of $v$, it holds that $\pa^*(y)\le v_2$, while
  $\pa^*(y)\in R[y]$. This implies that also $v_2\in R[y]$, and not
  $v_2\in T[y]$, hence $y$ cannot be a primal-dual ancestor of $v_1$
  and $v_2$.

  \emph{Second}, consider an ancestor $y$ of $v_1$ in $T^*$
  such that $v<y$. Because $v$ is an ancestor of $v_1$ in $T^*$, and
  $y$ is larger than $v$, $y$ is also an ancestor of $v$ in $T^*$. By
  lemma \ref{l.dualparent}, we know that $y\in R[v]$. This, in
  combination with $v_2\in T[v]$ implies that $v_2<y$, hence, $y$
  cannot be an ancestor of $v_2$ in $T$.
\end{proof}

\noindent For the following theorem, let
\begin{equation*}
  \dep_T(v_1,v_2):=\min\{\dep_T(y)\mid v_1\le y\le v_2\}
\end{equation*}
be the minimal depth of nodes between (and including) $v_1$ and $v_2$.

\begin{theorem}
  \label{t.primaldualchar}
  Let $v_1,v_2\in T\setminus \{r\}$ such that $v_1<v_2$. It
  holds that
  \begin{equation}
    \label{eq.primaldualchar}
    \pda(v_1,v_2) = \max_{<}\{v_1\le x\le v_2\mid \dep_T(x) = \dep_T(v_1,v_2)\}
  \end{equation}
  That is, according to depth-first traversal order in $T$, the primal-dual ancestor is
  the greatest node whose $T$-depth is minimal among all nodes between
  (and including) $v_1$ and $v_2$.
\end{theorem}

\noindent The proof is based on the following lemma:

\begin{lemma}
  \label{l.pardep}
  Let $v<w$ such that $w\in T[\pa^*(v)]$. Then it holds that
  \begin{equation}
    \label{eq.pardep}
    \dep_T(v,w) = \dep_T(\pa^*(v))
  \end{equation}
\end{lemma}

\noindent See 
Appendix \ref{app.pda}  
for a proof of lemma \ref{l.pardep}
and then theorem \ref{t.primaldualchar}.\\

\noindent Note immediately that theorem \ref{t.primaldualchar} implies that $v$
can be found in O(1) runtime, by performing a range minimum query on
the excess array $D$ of $\BP(T)$, defined by
$D[x]:=\rank_{1}(x)-\rank_{0}(x)$ where $\rank$ refers to
$\BP(T)$. Since $D[x+1]-D[x]\in\{-1,+1\}$, an RMQ on $D$ means
performing a $\pm 1$-RMQ, for which convenient solutions exist
\cite{Bender00}.\\

\noindent \textbf{Re-interpretation of RMQ's.}  Because it was shown
\cite{Fischer2011}, that the node in the 2D-Min-Heap $\cameraready{T[A]}$ that
corresponds to the solution of $\rmq_A(i,j)$ is given by the right
hand side of \eqref{eq.primaldualchar}, theorems \ref{t.pda} and
\ref{t.primaldualchar} allow for a reinterpretation of an RMQ query
$\rmq_A(i,j)$ on an array $A$ (without going into details here,
because the proof is an easy exercise based on collecting facts
from here, \cite{Fischer2011} and \cite{Ferrada2017}).

\begin{enumerate}
\item Determine the node $v$ in $\cameraready{T[A]}$ corresponding to $i$.
\item Determine the node $w$ in $\cameraready{T[A]}$ corresponding to $j$.
\item Determine $\pda(v,w)$ in $\cameraready{T[A]}$; return the corresponding index $i_o$.\\
\end{enumerate}

\noindent \textbf{Re-interpretation and improvement of Minimal Length
  Interval Queries (MLIQ).}  To illustrate the potential practical
benefits of our treatment, we further revisit the problem of
\emph{minimal length interval queries (MLIQ)}.  The improvements we
will be outlining are similar in spirit to the ones delivered in
\cite{Ferrada2017}. However, based on
our results, they are considerably more convenient to obtain.\\



\begin{problem}[MLIQ]
  \label{p.mininterval}
  Let $([a_i,b_i])_{i\in \{1,...,n\}}, a_i,b_i\in\mathbb{N}$ such that
  $a_i\le b_i$ for all $i\in \{1,...,n\}$ and
  $a_i < a_j$ and $b_i<b_j$ for $i<j$.
  \begin{itemize}
  \item \textbf{Input}: $(a,b)$ such that $a<b$
  \item \textbf{Output}: The index $i_0$ such that $[a_{i_0},b_{i_0}]$
    is the shortest interval that contains $[a,b]$, if such an
    interval exists.
  \end{itemize}
\end{problem} 

This problem makes part of other relevant problems, for example the
\emph{shortest unique interval} problem. In this context, a solution
for the MLIQ problem was presented in \cite{Hu2014} that requires
$O(b_n\log b_n)$ space to answer the query in $O(1)$ time. Therefore,
the following strategy was suggested. 

Let $l_i:=|b_i-a_i+1|$ be the length of the $i$-th interval,
$A:=[l_1,...,l_n]$ and $\cameraready{T[A]}$ the corresponding 2D-Min-Heap.

\begin{enumerate}
\item $i_{min}:=\min\{i\mid b_i>b\}, i_{max} := \max\{i\mid a_i<a\}$;
  if $i_{max}$ \cameraready{$<$} $i_{min}$ \textbf{output} 'None'.
\item Determine nodes $v,w\in \cameraready{T[A]}$ corresponding to
  $i_{min},i_{max}$.
\item Determine $\pda(v,w)\in \cameraready{T[A]}$; \textbf{output} its index.
\end{enumerate}

The solution presented in \cite{Hu2014} can immediately be improved by
employing bitmaps for the first step (which, according to
\cite{Raman2007}, requires $O(n\log(b_n/n))+o(b_n)$ space). Steps 2 and 3
then reflect an ordinary RMQ, which can be dealt with following
\cite{Ferrada2017}. In terms of query counts, Step 1 reflects two $\rank$
queries, while the resulting RMQ, following \cite{Ferrada2017}, requires
two $\select$'s, one $\pm 1$-$\rmq$, and one $\rank$.\\

If $|a_i-a_{i-1}|,|b_i-b_{i-1}|$ are in $O(\log n)$ (which applies for
several important applications), further improvements can be made
based on suggestions made in \cite{Tsur15} for BP representations of
trees with weighted parentheses. For that, we construct $T_a=\cameraready{T[A]}$ and
$T_b=\ova{\cameraready{T[A]}}$. We then assign weights $w_{a,i}:=|a_i-a_{i-1}|$ to
$i+1$-st opening parenthesis in $T_a$, whereas in $T_b$ we assign
$w_{b,i}:=|b_i-b_{i-1}|$ to the $i$-th closing parentheses (where
$a_0=b_0=0$; we recall that the number of non-root nodes in $\cameraready{T[A]}$ is
$n$). When aiming at running queries presented in \cite{Tsur15}, this
requires $2n\log\log n + o(n)$ bits of space, an improvement over
$O(n\log(b_n/n))+o(b_n)$ for the above, naive approach. Following
\cite{Tsur15}, let $\bpselect_{w_a,0}(a),\bpselect_{0,w_b}(b)$ be
defined by selecting the largest index in the balanced parenthesis
vector such that adding up all weights attached to opening parentheses
($w_a$) is at most $a$, or adding up all weights attached to closing
parentheses ($w_b$) is at most $b$. We can then run

\begin{enumerate}
\item $w:= \bpselect_{w_a,0}(a)$ in $T_a$ and
  $v:= 2n-\bpselect_{0,w_b}(b)+3$ in $T_b$; if $v>w$ \textbf{output}
  'None'
\item Determine $\pda(v,w)\in T_a$; \textbf{output} its index.
\end{enumerate}

In comparison to the naive approach from above, this makes two
$\bpselect$ queries, instead of two $\rank$'s and two $\select$'s.
The \emph{decisive trick} is to place $a$ and $b$ directly into
$\cameraready{T[A]}$, which avoids determining indices $i_{min},i_{max}$ first,
which subsequently need to be placed. Beyond the improvements in terms
of space and query counts, we argue that this solution reflects all
symmetries inherent to the MLIQ problem in a particularly compact
manner.

\section{Relating BP and DFUDS representations}
\label{sec.bpdfuds}

We will use the following construction to set up a tree induction for proving our main theorem.

\begin{definition}[Tree joining operation] \label{def:treejoining}
Let $T_1$ and $T_2$ be two trees, let $r_2$ be the root of $T_2$, $rmc_{T_2}(r_2)$ needs to exist and be a leaf. The notation ${T_1 \curvearrowright T_2}$ will denote a new tree formed by taking $T_2$ and inserting the children of the root of $T_1$ as children of the rightmost child of the root of the new tree.
Extend this operation to $n$ trees ${T_1,\ldots,T_n}$ where $T_2,\ldots,T_n$ all satisfy the same property as $T_2$ above, in the following way: ${T_1 \curvearrowright T_2 \curvearrowright T_3} = {(T_1 \curvearrowright T_2) \curvearrowright T_3}$ and so on, $${T_1 \curvearrowright T_2 \ldots \curvearrowright T_n} = {((\ldots((T_1 \curvearrowright T_2) \curvearrowright T_3) \curvearrowright \ldots) \curvearrowright T_n)}.$$ 
\end{definition}

\begin{observation} \label{obs:single_child} Let $T$ be a tree such that its root $r$ has a single child $c$ (that may or may not be a leaf). Then in $T^*$, by Rule 1b, $rmc_{T^*}(r)=c$ and is a leaf.\end{observation}

The following Lemma (proven in 
Appendix~\ref{app.bpdfuds_props})  

relates the dual tree to the tree joining operation. We will use the  $r \rightarrow T$ notation to denote a new tree formed by adding a new root $r$ as a parent of the root of $T$.

\begin{lemma}
Let $T$ be a tree consisting of a root $r$ and $n\geq 1$ subtrees $A_1, A_2, \ldots, A_n$ as children. When $n=1$, $T^*$ is $(r\rightarrow A_1)^*$. When $n\geq 2$, $T^*$ is $(r\rightarrow A_1)^* \curvearrowright   (r\rightarrow A_2)^* \curvearrowright  \ldots \curvearrowright (r\rightarrow A_n)^*$. 
\label{lem:dual-decomposition}
\end{lemma}

We are now ready to prove Theorem~\ref{thm:bpdfuds}. Parentheses in $\BP$ and $\DFUDS$ representations will be denoted by $\underline{(}$ and $\underline{)}$ to avoid confusion with usual mathematical parentheses. Recall that we use $\ova{s}$ to mirror a string $s$ of parentheses, e.g. $\ova{\underline{(()}} = \underline{())}$ and $\ova{\underline{)()}} = \underline{()(}$. 

\begin{proof}[Proof of Theorem~\ref{thm:bpdfuds}] Let $T$ be a tree with $n$ subtrees $A_1,\ldots,A_n$. It is clear that $\BP(T) = \underline{(} \BP(A_1) \BP(A_2) \ldots \BP(A_n) \underline{)}$. 
  Observe that for two trees $T_1$ and $T_2$ with roots $v_1$ and $v_2$, and  where $rmc_{T_1}(v_1),rmc_{T_2}(v_2)$ both exist and are leaves, $$\DFUDS(T_1\curvearrowright T_2) = \underline{(}\DFUDS(T_2 \setminus rmc_{T_2}(v_2))\DFUDS(T_1 \setminus rmc_{T_1}(v_1)) \underline{)}.$$ 
  
  In fact, one can show recursively that such a decomposition can be extended to $T_1\curvearrowright\ldots\curvearrowright T_n$. We will now prove the theorem with a tree structural induction.
  Observe that for a tree $T$ of depth 1 (a single root node), 
  $$\BP(T) = \underline{()}=\DFUDS(T^*)= \ova{\DFUDS(T^*)}.$$
  
  Now, assume the theorem equality is true for trees of depth $i$ and we will show it for trees of depth $i+1$. A tree $T$ of depth $i+1$ can be decomposed into a root node $r$ and $n$ subtrees $A_1,\ldots,A_n$ that are all of of depth $\leq i$ with roots $a_1,\ldots,a_n$. 
  Using Lemma~\ref{lem:dual-decomposition}, 
  $$\DFUDS(T^*) = \DFUDS((r\rightarrow A_1)^* \curvearrowright   (r\rightarrow A_2)^* \curvearrowright  \ldots \curvearrowright (r\rightarrow A_n)^*).$$
  By the recursive decomposition that we observed above, and using Observation~\ref{obs:single_child} stating that the rightmost child of $r$ in $(r\rightarrow A_i)^*$ is a leaf,
  $$\DFUDS(T^*) = \underline{(} 
  \DFUDS((r \rightarrow A_n)^* \setminus \{a_n\})\ldots 
\DFUDS((r \rightarrow A_1)^* \setminus \{a_1\})
\underline{)}.$$
Observe that we can take each DFUDS term in the expression above and wrap it around parentheses, i.e. $\underline{(}\DFUDS((r \rightarrow A_i)^* \setminus \{a_i\}\underline{)}$ which is equal to $\DFUDS((r \rightarrow A_i)^*)$.
Furthermore, note the following identity: $\DFUDS((r \rightarrow A_i)^*) = \underline{(}\DFUDS(A_i^*)\underline{)}$. 
And by inductive hypothesis, 
$
\DFUDS(A_i^*) = \ova{\BP(A_i)},
$
thus $\DFUDS((r \rightarrow A_i)^* \setminus \{a_i\}) = \ova{\BP(A_i)}$.
Hence,
$$\ova{\DFUDS(T^*)} = \underline{(}\BP(A_1)\ldots \BP(A_n)\underline{)} = \BP(T).$$ \end{proof}

\noindent \textbf{Proving \eqref{eq.motivation} from the
  Introduction.} Eventually, we also realize that
$\BP(\ova{T})=\ova{\BP(T)}$ and also
$\DFUDS(\ova{T})=\ova{\DFUDS(T)}$, both of which is straightforward
[$\star$]. Using this in combination with theorems \ref{t.dualdual} and
\ref{thm:bpdfuds}, we obtain
\begin{equation*}
  \DFUDS(\cameraready{T[A]}) \stackrel{[\star]}{=} \ova{\DFUDS(\ova{\cameraready{T[A]}})} \stackrel{Th.\ref{t.dualdual}}{=} \ova{\DFUDS((\ova{\cameraready{T[A]}}^*)^*)} \stackrel{Th.\ref{thm:bpdfuds}}{=} \BP(\ova{\cameraready{T[A]}}^*) \stackrel{D.\ref{d.reversedual}}{=} \BP(\widehat{\cameraready{T[A]}})
\end{equation*}
which establishes equation \eqref{eq.motivation} from the introduction.\\[2ex]


\textbf{Conclusive Remarks.}
In summary, we have provided a framework that unifies $\BP$ and $\DFUDS$.
From a certain point of view, we have pointed out that neither should $\BP$ based approaches have advantages over $\DFUDS$ based approaches, nor vice versa.
As an exemplary perspective of our framework, $\BP$ based treatments such as \cite{Navarro2014,Tsur15} might have an easier grasp of the advantages that $\DFUDS$ based approaches bring along. 
Finally, we consider it interesting future work to also characterize trees that put $\BP$ and/or $\DFUDS$ based representations into context with $\LOUDS$ based representations.

\bibliography{dualtree}


\appendix

\section{The simpler query from \cite{Ferrada2017} also works in \cite{Fischer2011}: direct proof}
\label{app.A}

In the following, we identify nodes of $T$ with the closing parenthesis
that represent them in $\DFUDS$, that is $v = \CP(v)$. \cameraready{Recall that $D$ is the array defined in Section~\ref{ssec.motivation}.}

\begin{lemma}
  \label{lem.dfuds1}
  Let $v_2$ the immediate right sibling of $v_1$. Then, in $\DFUDS(T)$,
  \begin{equation}
    \label{eq.dfudslem1}
    D[v_2] = D[v_1]-1
  \end{equation}
\end{lemma}

\begin{proof}
  Given \eqref{eq.dfudslem1}, we show that all parentheses between
  $v_1$ and $v_2$ are elements of $T[v_1]$, the subtree hanging off
  (but here not including) $v_1$. In other words, we will show that
  \begin{equation}
    x\in T[v_1] \quad\text{if and only if}\quad v_1<v<v_2
  \end{equation}

  For ``$\Rightarrow$'', the first case is that $v$ represents a
  closing parenthesis. Then the claim follows because closing
  parentheses come in depth-first traversal order, hence $v$
  comes after $v_1$, and before $v_2$. The second case is that
  $v$ represents an opening parenthesis. So, by $\DFUDS$ principles,
  the first closing parenthesis to the left of $v$ refers to $v$'s parent,
  which is either itself a member of $T[v_1]$ or $v_1$ itself. In both
  cases, $x$ comes after $v_1$ and before $v_2$.

  For ``$\Leftarrow$'', the case of $v$ being a closing parenthesis
  implies the claim because of the depth-first traversal order. The
  case of $v$ being an opening parenthesis requires to look at the
  first closing parenthesis $u$ to the left, which refers to the
  parent of $v$. We obtain $\{v\}\subset T[u]\subset T[v_1]$, because
  $u$ either is a descendant of $v_1$ or $v_1$ itself.
\end{proof}

\begin{lemma}
  \label{lem.dfuds2}
  Let $v_2$ be the rightmost child of $v_1$. Then, in $\DFUDS(T)$,
  \begin{equation*}
    D[v_2] = D[v_1]
  \end{equation*}
\end{lemma}

\begin{proof}
  By $\DFUDS$ logic, $\OP(v_2)$ directly follows $v_1$. Further, again
  by $\DFUDS$ logic, the parentheses between $\OP(v_2)$ and $v_2$ are
  exactly the members of subtrees of all children of $v_1$, but $v_2$.
  That is, we are facing the following situation:
  
  \begin{equation}
    \underset{v_1}{\big)}\;\underset{\OP(v_2)}{\big(}\; \underset{T[v_1]\setminus T[v_2]}{\underbrace{\big( ...............\big)}}\quad \underset{v_2}{\big)}\quad\underset{T[v_2]}{\underbrace{\big(...... }}
  \end{equation}
  So, $D[v_2-1] = D[\OP(v_2)]$, and further $D[v_2] = D[v_2-1]-1$ and
  $D[v_1] = D[\OP(v_2)]-1$, which together implies
  \begin{equation}
    D[v_2] = D[v_2-1]-1 = D[\OP(v_2)]-1 = D[v_1]
  \end{equation}
  
\end{proof}

To provide a direct proof of the fact that Ferrada and Navarro's query
also works for Fischer and Heun, we have to show that in $\DFUDS(\cameraready{T[A]})$,
\begin{equation*}
  \rank_)(\open(w_1)) = i
\end{equation*}
is equivalent to
\begin{equation*}
  \rmq_D(\select_)(i),\select_)(j)) = \select_)(i)
\end{equation*}

Recalling that $\rmq_D$ refers to the \emph{leftmost minimum} in the
array $D$, where $D[x]=\rank_1(x)-\rank_0(x)$ for a parenthesis
$x\in\DFUDS(\cameraready{T[A]})$, we have to prove the following technical lemma.

\begin{lemma}
  In $\DFUDS(T)$, the following two statements are equivalent:
  \begin{itemize}
  \item[]\,$(i)$ 
    \begin{equation}
      \text{For all}\quad x\in[\select_0(i+1),\select_0(j)]:\quad D[\select_0(i)] \le D[x]
    \end{equation}
  \item[]\,$(ii)$
    \begin{equation}
      \rank_0(\open(\rmq_D(\select_0(i+1),\select_0(j)))) = i
    \end{equation}
  \end{itemize}
\end{lemma}

\begin{proof}
  $(i)\Rightarrow (ii)$: By lemmata \ref{lem.dfuds1} and
  \ref{lem.dfuds2}, we know that the first parenthesis $x$ right of
  $\select_0(i)$ where $D[x] < D[\select_0(i)]$ is the sibling right
  of the node $v$ represented by $\select_0(i)$. Hence by $\DFUDS$
  logic, $(i)$ implies that all $x\in[\select_0(i+1),\select_0(j)]$
  refer to descendants of $v$. Again using lemmata \ref{lem.dfuds1}
  and \ref{lem.dfuds2}, we infer that
  $w_1:=\rmq_D(\select_0(i+1),\select_0(j))$ refers to the closing
  parenthesis of the rightmost child of $v$ among the children of $v$
  showing in $[\select_0(i+1),\select_0(j)]$---note that there is at
  least one, because $\select_0(i+1)$ refers to the leftmost child of
  $v$. So, $\open(w_1)$ is one of the opening parentheses directly
  following $\select_0(i)$, that is
  $\rank_0(w_1)=\rank(\select_0(i))=i$.

  $(ii)\Rightarrow (i)$: If $(ii)$ applies,
  $w_1:=\rmq_D(\select_0(i+1),\select_0(j))$ is a closing parenthesis
  whose opening counterpart follows $\select_0(i)$ without any closing
  parenthesis in between. That is, $w_1$ represents one of the
  children of $\select_0(i)$. From lemmata \ref{lem.dfuds1} and
  \ref{lem.dfuds2}, we infer that $D[\select_0(i)]\le D[w_1]$ with
  equality if and only if $w_1$ represents the rightmost child of
  $\select_0(i)$. Because $D[w_1]$ was selected as the minimum among
  the $D[x], x\in [\select_0(i+1),\select_0(j)]$, we obtain
  \begin{equation}
    D[\select_0(i)]\le D[w_1] \le D[x]
  \end{equation}
  
\end{proof}
 

  
  \section{Proof of Lemma \ref{l.dualparent} and Corollaries \ref{c.suborder}, \ref{c.dualsubtree}}
  \label{app.applemma1}

  \textbf{Proof of Lemma \ref{l.dualparent}.}
  We consider the three different cases that correspond to Rules 1b,
  3 and 2 (in that order).\\
  
  Ad \textbf{Rule 1b}: If $v$ is the rightmost child of the root, all nodes
  that follow $v$ in depth-first traversal order are in the subtree
  $T[v]$ of $v$, so $R[v]$ is the empty set.\\
  
  Ad \textbf{Rule 3}: If $v=ils_T(u)$, then $u$, in depth-first traversal order,
  is the first node following nodes in $T[v]$, the subtree of $v$, that is,
  $u$ is the smallest node in $R[v]$, so $\pa^*(v)=u$.\\

  Ad \textbf{Rule 2}: Here, $v=rmc_T(u)$. We lead the proof by induction on
  $\dep_T(v)$, where the start, $\dep_T(v)=1$, is given by the already proven
  case of Rule 1b. Let $i\ge 1$ and $\dep_T(v)=i+1$. As $\pa_T(v)=u$, it
  holds that $\dep_T(u)=i$, so by the induction assumption, in combination
  with $\pa^*(u)=\pa^*(v)$, we obtain
  
  \begin{equation}
    \label{eq.rvinduct}
    \pa^*(v) =
    \begin{cases}
      r & R[u] = \emptyset\\
      \min R[u] & R[u] \ne \emptyset
    \end{cases}.
  \end{equation}
  
  Case 1, $R[u]=\emptyset$: $v$ being a child of $u$ implies
  $T[v]\subset T[u]$. The assumption $R[v]\ne\emptyset$ would imply
  the existence of a node $x$ right of $v$. Since $R[u]=\emptyset$, we
  obtain by \eqref{eq.subrelations} that $R[v]\subset T[u]$, so
  $x\in T[u]$. The combination of $x$ being right of $v$ and being
  part of the subtree $T[u]$ rooted at the parent $u$ of $v$ implies
  the existence of a right sibling of $v$, which is a contradiction to
  $v=rmc_T(u)$.
  
  Case 2, $R[u]\ne\emptyset$: Let $x:=\pa^*(v) = \min R[u]$ and let
  $y:=\min R[v]$. We need to show that $x = y$. Because of
  \eqref{eq.subrelations}, we know that $y\le x$. The assumption
  $y<x$, however, implies the existence of a node right of $v$ that
  lies in $T[u]$, which again contradicts $v=rmc_T(u)$.\qed\\


  \begin{corollary}
    \label{c.suborder}
    Let $u$ be an ancestor of $v$ in $T$. Then
    \begin{equation*}
      \pa^*(v) \le_T \pa^*(u).
    \end{equation*}
  \end{corollary}

  \begin{proof}
    This follows from lemma \ref{l.dualparent} in combination with
    \eqref{eq.subrelations}, where the latter states that
    $R[u]\subset R[v]$, making $\min R[v] \le_T \min R[u]$.
  \end{proof}

  The next corollary is an immediate consequence of corollary
  \ref{c.suborder}.

  \begin{corollary}
    \label{c.dualsubtree}
    Let $w=\pa^*(v)$ for $v\in T\setminus \{r\}$. Then
    \begin{equation*}
      T[v] \subset T^*[w].
    \end{equation*}
  \end{corollary}

  \begin{proof}
    Let $v_2\in T[v]$, that is $v$ is an ancestor of $v_2$ in $T$. We
    have to show that $w$ is an ancestor of $v_2$ in $T^*$. From
    corollary \ref{c.suborder}, we know that
    $\pa^*(v_2) \le_T \pa^*(v)$, so, because $\pa^*(v)$ is the first
    node in $T$-order not in $T[v]$, either $\pa^*(v_2)=\pa^*(v)$ or
    $\pa^*(v_2)\in T[v]$. In the first case, we are done. In the
    second case, we repeat this argument by replacing $v_2$ with
    $\pa^*(v_2)$ (formally: induction on the number of nodes between
    $v_2$ and $\pa^*(v)$ in $T$-order) to conclude the proof.
  \end{proof}

     \section{Proof of Theorem \ref{t.dualminheap}.}
  \label{app.dualminheap}

  Here, we provide a proof for our motivating theorem
  \ref{t.dualminheap}.  Le $A=A[1,n]$ be an array. For additional
  clarity, we require $S_A:=\{A[1],...,A[n]\}$ to form a totally ordered set,
  implying that for all $i\ne j, 1\le i,j\le n$, either $A[i]<A[j]$ or
  $A[j]<A[i]$, and note that this requirement can easily be overcome in
  applications. As before, $\overleftarrow{A}$ is the reversal $A$,
  given by $\overleftarrow{A}[i]:=A[N-i+1]$.
  
  We will deal with two orderings in the following, namely, the one on
  the list indices $\{1,...,N\}$ and the one on the set
  $S_A:=\{A[1],...,A[N]\}$. If distinction is required, we write $<_A$
  for the former and $<_S$ for the latter. 
  
  \begin{definition}
    \label{d.minheap}
    (from \cite{Fischer2011}) The \emph{2D-Min-Heap} $\cameraready{T[A]}$ for an array $A[1,N]$
    is a rooted, ordered tree where, first,
    \begin{equation*}
      V\setminus r = \{A[1],...,A[N]\}.
    \end{equation*}
    Edges are determined by way of iteratively determining the
    parent of $A[m+1]$ in $T_{A[1,m]}$:
    \begin{enumerate}
    \item The parent of $A[1]$ is $r$.
    \item Let $T_{A[1,m]}$ be already constructed. Then
      $A[m+1]=rmc(A[k])$ where $k:=\max_{<_A}\{l\in\{1,...,m\}\mid A[l] <_S A[m+1]\}$.
    \end{enumerate}
    That is, $A[m+1]$ is appended as the rightmost child to the
    rightmost element in $A[1,m]$ that is smaller than $A[m+1]$.
  \end{definition}
  
  We make some observations leading to a characterization of the
  depth-first traversal order on $\cameraready{T[A]}$, all of which are straightforward
  (and well known).
  
  \begin{observation}
    \label{o.minheap1}
    Let $A[k]<_SA[i]$ for all $k+1\le i\le k+l$. Then $A[i]\in \cameraready{T[A]}[A[k]]$
    for all $k+1\le i\le k+l$
  \end{observation}

    \noindent \textbf{Proof of Observation \ref{o.minheap1}.} One
  obtains this insight by induction on $i$. By construction of $\cameraready{T[A]}$,
  we obtain that $A[k+1]$ is appended to $T_{A[1,k]}$ as rightmost
  child of $A[k]$, which makes the start. Consider $A[k+i+1]$ for
  $1\le i\le l-1$. Since $A[k]<A[k+i+1]$, the parent of of $A[k+i+1]$
  is one of the $A[k],...,A[k+i]$. By the induction assumption, that
  parent is an element of $\cameraready{T[A]}[A[k]]$, so also
  $A[k+i+1]\in \cameraready{T[A]}[A[k]]$. \qed\\

  \begin{observation}
    \label{o.minheap2}
    Let $A[m]\in \cameraready{T[A]}[A[k]]$. Then $A[k] <_S A[m]$.
  \end{observation}

    \noindent \textbf{Proof of Observation \ref{o.minheap2}.} This
  insight is an immediate consequence following from the fact that a
  node $A[j]$ is greater than its parent $A[i]$, that is,
  $A[j]<_SA[i]$. \qed\\
 
  \noindent With these observations at hand, we can prove the following
  (well-known, intuitively straightforward) lemma.
  
  \begin{lemma}
    \label{l.minheapdft}
    The depth-first traversal order $<_{\cameraready{T[A]}}$ on $\cameraready{T[A]}$ coincides with the order
    $<_L$ on $\{1,...,N\}$. That is, for $i,j\in \{1,...,N\}$
    \begin{equation*}
      i<j \quad \text{ if and only if } \quad A[i]<_{\cameraready{T[A]}}A[j].
    \end{equation*}
  \end{lemma}

  \noindent \textbf{Proof of Lemma \ref{l.minheapdft}.}
  Let $1\le i\le N-1$. It suffices to show that $A[i+1]$ comes after
  $A[i]$ in depth-first traversal order on $\cameraready{T[A]}$. Therefore, let
  $A[k]:=\pa_{\cameraready{T[A]}}(A[i+1])$. If $k=i$, hence $A[k]=A[i]$, we are done.
  If not, consider all nodes $A[k+1],...,A[i]$ between $A[k]$ and
  $A[i+1]$. By construction of $\cameraready{T[A]}$, we know that $A[i+1]<_SA[j]$
  for all $k+1\le j\le i$, while $A[k]<_SA[i+1]$, which implies that
  $A[k]<_SA[j]$ for all $k+1\le j\le i$. So, by observation
  \ref{o.minheap1}, all $A[j]\in T[A[k]]$ for $k+1\le j\le i$, implying
  in particular that $A[i]\in T[A[k]]$. Since during construction of
  $\cameraready{T[A]}$ $A[i+1]$ is appended as rightmost child of $A[k]$ after $A[i]$
  had been appended, $A[i+1]$ comes after $A[i]$ in depth-first
  traversal order on $\cameraready{T[A]}$. \qed\\
  
  We are now in position to prove theorem \ref{t.dualminheap}.
  
  \begin{theorem}
    Let $A[1,N]$ be an array of (mutually different) numbers and
    let $\overleftarrow{A}:=[A[N],...,A[1]]$ be the reversal of it. Then
    \begin{equation}
      (\cameraready{T[A]})^* = \cameraready{T[\overleftarrow{A}]}.
    \end{equation}
  \end{theorem}
  
  
  \begin{proof}
    By applying lemma \ref{l.minheapdft} for $\cameraready{T[\overleftarrow{A}]}$,
    the depth-first traversal order on nodes in
    $\cameraready{T[\overleftarrow{A}]}\setminus r$ agrees with the reverse order on
    $\{1,...,N\}$. By applying lemma \ref{l.minheapdft} for $\cameraready{T[A]}$ and
    combining it with lemma \ref{l.dualorder} for $\cameraready{T[A]}$, we see that the
    depth-first traversal order on $\cameraready{T[\overleftarrow{A}]}$ agrees with
    that on $\cameraready{T[A]}^*$.\\
    
    It remains to show that the parent of $A[k]$ in $(\cameraready{T[A]})^*$ agrees
    with the parent of $A[k]$ in $\cameraready{T[\overleftarrow{A}]}$. We recall
    lemma \ref{l.dualparent} and know that $\pa_{(\cameraready{T[A]})^*}(A[k])=r$ if
    $R_{\cameraready{T[A]}}[A[k]]=\emptyset$ (\emph{first case}\/) and
    $\pa_{(\cameraready{T[A]})^*}(A[k])=\min_{<_{\cameraready{T[A]}}}R_{\cameraready{T[A]}}[A[k]]$ if
    $R_{\cameraready{T[A]}}[A[k]]$
    is not empty (\emph{second case}\/).\\
    
    For the \emph{first case}, we are done if $A[k]=A[N]$, because then
    $\pa_{\cameraready{T[\overleftarrow{A}]}}(A[N])=r$ by construction of
    $\cameraready{T[\overleftarrow{A}]}$. If not, $R_{\cameraready{T[A]}}[A[k]]=\emptyset$
    translates into $A[i]\in T[A[k]]$ for $k+1\le i\le N$. That is, when
    constructing $\cameraready{T[\overleftarrow{A}]}$, there is no node in
    $T[[A[N],...,A[k+1]]]$ that is smaller than $A[k]$, in which case
    $A[k]$ is appended to $T[[A[N],...,A[k+1]]]$ as the rightmost child
    of the root, so also
    here $\pa_{\cameraready{T[\overleftarrow{A}]}}(A[N])=r$.\\
    
    In the \emph{second case}, we consider
    $A[l]:=\min_{\cameraready{T[A]}}R_{\cameraready{T[A]}}[A[k]]$, the parent of $A[k]$ in $(\cameraready{T[A]})^*$. So,
    by definition of $R_{\cameraready{T[A]}}[A[k]]$, we have that $A[i]\in \cameraready{T[A]}[A[k]]$ for
    all $i: k+1\le i\le l-1$. So, by observation \ref{o.minheap2},
    $A[k]<A[i]$ for all $i: k+1\le i\le l-1\; (\star)$.
    
    Furthermore, $A[l]\in R_{\cameraready{T[A]}}[A[k]]$ translates into the fact that
    $A[l]$, during the construction of $\cameraready{T[A]}$, was not appended to a child
    of any of the nodes in $\cameraready{T[A]}[A[k]]$, so $A[l]<A[i]$ for all nodes
    $A[i]\in \cameraready{T[A]}[A[k]]$, which implies in particular that $A[l]<A[k]\; (\star\star)$.
    
    \begin{sloppypar}
      So, when appending $A[k]$ to $T[[A[N,...,A[k+1]]]$ during the
      construction of $\cameraready{T[\overleftarrow{A}]}$, combining $(\star)$ and
      $(\star\star)$ yields that $A[l]$ was found to be the rightmost element
      in $[A[N,...,A[k+1]]$ that was smaller than $A[k]$, which agrees
      with the definition of the parent of $A[k]$ in
      $\cameraready{T[\overleftarrow{A}]}$. 
    \end{sloppypar}
    
  \end{proof}

 \section{Proofs of Lemmata \ref{l.dualorder}, \ref{l.dualleft}, \ref{l.siblingsubtree}}
  \label{app.lemma3}
  
  \textbf{Proof of Lemma \ref{l.dualleft}.}
  The edge $(w,v_2)$ in $T^*$ cannot be due to Rule 3, because $v_2$
  is not the immediate left sibling of $w$ in $T$, which would imply
  that $w\in T[v]$, which contradicts $w\in R[v]$ or $w$ being the
  root, which is established by $w=\pa^*(v)$ and lemma
  \ref{l.dualparent}.
  
  Note that, since $v$ is not the root, also Rule 1b does not apply.
  So the edge $(w,v_2)$ in $T^*$ must have come into existence by
  Rule 2. That is, $v_2=ils_{T^*}(v_3)$ where $v_2$ was the rightmost
  child of $v_3$ in $T$. We have $v_3\in T[v]$ and $\pa^*(v_3)=w$.
  We are done if $v_3=v$, because then $v_2$ is the immediate left
  sibling of $v$ in $T^*$. If not, we obtain the claim by induction
  on $\dep(v_2)-\dep(v)$. \qed\\

  \noindent \textbf{Proof of Lemma \ref{l.siblingsubtree}.}
  By Rule 3, we know that the immediate right sibling $v_2$ of $v_1$
  in $T$ is the parent of $v_1$ in $T^*$, in other words,
  $v_2=\pa^*(v_1)$. Corollary \ref{c.dualsubtree} implies that
  $T[v_1]\subset T^*[v_2]$, so we are done if $v_2=u_1$.  If not, then
  repeated application of Rule 3 implies that $u_1$ is an ancestor of
  $v_2$ in $T^*$. In other words, $v_2\in T^*[u_1]$, hence also
  $T^*[v_2]\subset T^*[u_1]$, which finally yields
  $T[v_1]\subset T^*[u_1]$, as claimed. \qed\\
  
  We can now proceed with proving lemma \ref{l.dualorder}.\\
  
  \noindent \textbf{Proof of Lemma \ref{l.dualorder}.}
  It suffices to show that $v<u$ implies $u<^*v$. If $v<u$,
  two different cases can apply, either $u\in T[v]$ or $u\in R[v]$.\\
  
  Ad $u\in T[v]$: Let $w:=\pa^*(v)$. Corollary \ref{c.dualsubtree}
  implies that $w$ is an ancestor of $u$. Since all ancestors of $u$
  in $T^*$ are greater than $v$ in terms of depth-first traversal
  order in $T$, we obtain the existence of a node $v_2>v$ such that
  $\pa^*(v_2)=w$ where, possibly, $v_2=u$ itself. We obtain the claim
  by applying lemma \ref{l.dualleft}.\\
  
  Ad $u\in R[v]$: Let $w$ be the least common ancestor of $v$ and $u$
  in $T$ and let $v_1, u_1$ be the children of $w$ such that $v\in T[v_1]$
  and $u\in T[u_1]$. By the prior case $u\in T[v]$, we know that $u<^*u_1$.
  Application of lemma \ref{l.siblingsubtree} then further yields that
  $u_1<^*v$, which implies the desired $u<^*v$. \qed
 
  \section{Proof of Theorem \ref{t.dualdual}.}
  \label{app.dualofdual}

  We first consider the case $R^*[v]=\emptyset$. Here, by depth-first
  traversal order in $T^*$, all ancestors $v_1<^* ... <^*v_l<^*v$ of
  $v$ (apart from the root $r$) and $v$ are rightmost children. By
  repeated application of Rule 3, we see that all nodes
  $v_1> ... > v_l > v$ are siblings in $T$, where $v$ is the
  leftmost. So, the parent of $v$ agrees with the parent of $v_1$,
  which is the rightmost child of the root $r$. By Rule 1b, we see
  that also in $T$, the parent of $v_1$ is the root $r$.\\

  We now consider the case $R^*[v]\ne\emptyset$.  First, $u<v$ implies
  $v<^*u$ using lemma \ref{l.dualorder}.  The assumption $u\in T^*[v]$
  implies that $v$ is an ancestor of $u$ in $T^*$, which is
  impossible, because lemma \ref{l.dualparent} then says that
  $v\not\in T[u]$. So $u\in R^*[v]$, and it remains to show that
  $u$ is the smallest node in $R^*[v]$, according to $<^*$.

  We assume the existence of $x\in R^*[v]$ that is smaller than $u$,
  and show that this leads to a contradiction. By lemma
  \ref{l.dualorder}, this implies that $u<x$, so either (1)
  $x\in T[u]$ or (2) $x\in R[u]$.

  Ad (1): For $x\in T[u]$, let $x_1$ be the child of $u$, such that
  $x\in T[x_1]$. If $x_1$ is a left sibling of $v$, we obtain
  $x\in T[x_1] \subset T^*[v]$ by lemma \ref{l.siblingsubtree}, a
  contradiction to $x\in R^*[v]$. If $x_1$ is a right sibling of $v$,
  we obtain $v\in T^*[x_1]$, again by lemma \ref{l.siblingsubtree},
  which implies $x_1<^*v$. Further, $x\in T[x_1]$ implies $x_1<x$, and
  further into $x<^*x_1$ by lemma \ref{l.dualorder}. Together, we
  obtain $x<^*v$, a contradiction to $x\in R^*[v]$.

  Ad (2): It remains to consider the case $x\in R[u]$. By depth-first
  traversal order in $T$, we have $u<v<x$, so in $T^*$, by lemma
  \ref{l.dualorder}, $x<^*v<^*u$.  This contradicts that
  $x\in R^*[v]$, which concludes the proof.
\qed

  \section{Proof of Proposition \ref{p.reversed}}
  \label{app.reversed}

  Note that $(a)-(c)$ are immediate. For $(d)-(g)$ note that rules
  for $\ova{T^*}$ basically reiterate the rules for the dual tree,
  while exchanging left with right.
  
  \begin{proof}
    Again, these are straightforward observations, obtained by reversing
    the order among the children of nodes in $T^*$, which yields, as one
    example, that in $\ova{T^*}$, the immediate left sibling $v$ of $u$
    in $T$ turns into the rightmost child of $u$ in $T^*$, which by
    reversing $T^*$ turns into the leftmost child of $T^*$, and so on.
    Computing $\ova{T}^*$, the dual of the reversed tree, we find that
    $(d)-(g)$ apply also for $\ova{T}^*$, just as for $\ova{T^*}$. Since
    one can show that $(d)-(g)$ are defining properties of $\ova{T^*}$,
    in analogy to the insight that rules 1-3 from Definition
    \ref{def.dualtree} give rise to the dual tree $T^*$ itself, we see
    that $\ova{T}^*$ and $\ova{T^*}$ must be identical.
  \end{proof}
  
  \section{Proof of Lemma \ref{l.pardep} and Theorem \ref{t.pda}.}
  \label{app.pda}

  In the following, we write $\dist_T(v,w)$ and $\dist_{T^*}(v,w)$ for
  the length of a minimum length path between $v$ and $w$ in $T$ and
  $T^*$, respectively. We will also write $y=(\pa)^i(v)$ and
  $y=(\pa^*)^i(v)$ if $y$ is the $i$-th ancestor of $v$ in $T$ or
  $T^*$, respectively.

  In the following, 'first', 'largest', and so on, refer to depth-first
  traversal order in $T$. When referring to depth-first traversal order
  in $T^*$, we will explicitly mention this.\\

  \noindent For the proofs, we recall that 
  \begin{equation*}
    \dep_T(v_1,v_2):=\min\{\dep_T(y)\mid v_1\le y\le v_2\}
  \end{equation*}
  is the minimal depth of nodes between (and including) $v_1$ and
  $v_2$. We then observe the following relationship for
  $v_1\le w\le v_2$:
  \begin{equation}
    \label{eq.transdep}
    \dep_T(v_1,v_2) = \min(\dep_T(v_1,w),\dep_T(w,v_2))
  \end{equation}

  \noindent \textbf{Proof of Lemma \ref{l.pardep}.}
  First, all nodes that follow $\pa^*(v)$ in depth-first
  traversal order until $w$ are in $T[\pa^*(v)]$, hence have depth
  greater than $\pa^*(v)$.  Second, by lemma \ref{l.dualparent}, all
  nodes between $v$ and $\pa^*(v)$ are members of $T[v]$, hence have
  greater depth than $v$. The dual parent $\pa^*(v)$ of $v$, by lemma
  \ref{l.dualparent} the first node in $R[v]$ following the nodes in
  $T[v]$ is either a right sibling of $v$ or a right sibling of one of
  the ancestors of $v$, all of which is a direct consequence of
  depth-first traversal order. Either way,
  $\dep_T(\pa^*(v))\le \dep_T(v)$.\qed\\

  \noindent \textbf{Proof of Theorem \ref{t.pda}.}
  Let $v:=\pda(v_1,v_2)$. We encounter the following situation: $v$ is
  the largest node smaller or equal to $v_2$ that is a $T^*$-ancestor
  of $v_1$. So, in particular, $\pa^*(v)>v_2$. By lemma
  \ref{l.dualparent}, $\pa^*(v)$ is the first node following $v$ that
  is not in the subtree $T[v]$ rooted at $v$. So, all nodes following
  $v$, until and including $v_2$ are in $T[v]$, hence have depth larger
  than $v$. That is,
  \begin{equation}
    \label{eq.primaldualchar1}
    \dep_T(v,v_2) = \dep_T(v),
  \end{equation}
 Further, again
  by lemma \ref{l.dualparent}, $\pa^*(y)$ is the first node following
  $y$ that is not in $T[y]$. Hence, by definition of depth-first order traversal,
  $\pa^*(y)$ is either the right sibling of $y$ or one of its ancestors
  (if $y$ is the rightmost child of its parent). Either way,
  \begin{equation}
    \label{eq.primaldualchar2}
    \dep_T(y,\pa^*(y)) = \dep_T(\pa^*(y))
  \end{equation}
  Let $i$ be such that $v=(\pa^*)^i(v_1)$ (*), that is, $v$ is the $i$-th
  ancestor of $v_1$ in $T^*$. Repeated application of \eqref{eq.transdep}
  yields
  \begin{multline}
    \label{eq.primaldualchar3}
    \dep_T(v_1,v_2)
    \stackrel{\eqref{eq.transdep}}{=}\\ \min(\dep_T(v_1,\pa^*(v_1)), \dep_T(\pa^*(v_1), (\pa^*)^2(v_1)),\\
    \ldots, \dep_T((\pa^*)^{i-1}(v_1), (\pa^*)^i(v_1)), \dep_T((\pa^*)^i(v_1),v_2))\\
    \stackrel{\eqref{eq.primaldualchar1}, \eqref{eq.primaldualchar2}}{=} \min(\dep_T(\pa^*(v_1)),\dep_T((\pa^*)^2(v_1)),
    \ldots, \dep_T((\pa^*)^i(v_1)))\\
    \stackrel{\eqref{eq.primaldualchar2},(*)}{=} \dep_T(v)
  \end{multline}
  so $v$ indeed achieves minimal depth among all nodes
  $v_1\le x\le v_2$.  Because of \eqref{eq.primaldualchar1}, all nodes
  following $v$ have depth larger than $v$, hence $v$ is also the
  largest node that minimizes the depth between (and including) $v_1$
  and $v_2$, which implies our claim. \qed
  
  \section{Proof of Lemma~\ref{lem:dual-decomposition}.}
  \label{app.bpdfuds_props}

The lemma requires to prove an equality between a dual tree and the tree-join of several dual trees.
In general, if $A$ is an induced subgraph of $T$ then one cannot always relate $A^*$ and $T^*$ in terms of inclusion (see e.g. Figure 1 with $A=T[u]$). First we will study more precisely which edges of $A^*$ are in $T^*$. 

\begin{definition}[Quasi-subtree] Let $T,A$ be trees such that $A$ is an induced subgraph of $T$. $A$ is a \emph{quasi-subtree} of $T$ if for any two nodes $u,v$ in $A$, $u=ils_A(v) \implies u=ils_T(v)$, and when $v$ is not the root of $A$, $u=rmc_A(v) \implies u=rmc_T(v)$.
\end{definition}

Observe that a subtree is also a quasi-subtree, but not the other way around.

\begin{lemma}
Let $A$ be a quasi-subtree of $T$.
All edges of $A^*$ that are non-incident to the root of $T$ are also present in $T^*$.
\label{lem:dual-edges-notroot}
\end{lemma}

\begin{proof}
 Let $r$ be the root of $T$. Edges of $A^*$ that are not incident to $r$ are either created by Rule 2 or by Rule 3 in Definition~\ref{def.dualtree}. 
First, consider the edges due to Rule 3. Let the edge $(u,v)\in A^*$ that arises from $v=rmc_{A^*}(u)$ for some $u\neq r$. Rule 3 was applied because $v=ils_{A}(u)$, and thus by hypothesis on $A$, $v=ils_{T}(u)$. Applying Rule 3 to $T$ yields that $v=rmc_{T^*}(u)$, hence $(u,v)\in T^*$.

Second, consider the edges due to Rule 2. Let $v$ be a node of $A^*$ and its parent $w\neq r$ in $A^*$, such that and $v=ils_{A^*}(u)$ for some $u$. We need to show that $(w,v)$ is in $T^*$ also. Let $u_0 = v$ and $u_1 = u, \ldots, u_n$ be all the right siblings of $v$ in $A^*$. We show by induction from $i=n$ to $i=0$ that $(w,u_i), n \geq i \geq 0$ is in $T^*$. 
For the base case ($i=n$), $u_n=rmc_{A^*}(w)$, thus $u_n = ils_A(w)$, thus by hypothesis on $A$, $u_n=ils_T(w)$, then $u_n=rmc_{T^*}(w)$. 
Now, assuming by induction that $(w,u_i)\in T^*$, consider the edge $(w,u_{i-1})$ where $u_{i-1} = ils_{A^*}(u_i)$. 
By Definition~\ref{def.dualtree}, $u_{i-1} = rmc_A(u_i)$, and since $w\neq r$, $u_i$ cannot be the root of $A$ therefore by hypothesis on $A$, $u_{i-1} = rmc_T(u_i)$, 
hence applying Definition~\ref{def.dualtree} to $T$ yields that $u_{i-1} = ils_{T^*}(u_i)$, hence that $(w,u_{i-1})\in T^*$. Therefore the induction is complete, and $(w,v)\in T^*$. 
\end{proof}

The next lemma will make use of the following observation.

\begin{observation} Let $T_1$ and $T_2$ be two trees (having roots $r_1,r_2$) such that $r_2$ has a single child. From Observation~\ref{obs:single_child}, $T_1^* \curvearrowright T_2^*$ exists. Its edges can be partitioned into three types: (i) $\{(rmc_{T_1^* \curvearrowright  T_2^*}(r_2),v) \mid pa_{T_1^*}(v) = r_1 \}$ (edges that were inserted from the children of the root of $T_1^*$), (ii) $\{(r_2,w) \mid pa_{T_2^*}(w) = r_2 \}$ (edges which were incident to $r_2$ in $T_2^*$), and (iii) edges that are neither incident to $r_2$ nor $rmc_{T_2^*}(r_2)$.
\label{obs:treejoin-edges}
\end{observation}

\begin{lemma}
Let $T$ be a tree rooted at $r$. Let $T_2$ be the subtree of $T$ that is rooted at $rmc_{T}(r)$, and $T_1 = T \setminus T_2$. Then $T^*$ is $T_1^* \curvearrowright (r\rightarrow T_{2})^*$.
\label{lem:two-trees-join-is-dual}
\end{lemma}
\begin{proof}
We will set $T'=T_1^* \curvearrowright (r\rightarrow T_{2})^*$. Observe that $T'$ has exactly the same set of nodes as $T$ and $T^*$, as the extra $r$ node in $(r\rightarrow T_2)^*$ is deleted by the tree joining operation. Therefore to show equality of two trees having the same number of nodes, one only needs to show an edge inclusion. We will show that the edges of $T'$ are in $T^*$.
We will consider the three types of edges in $T'$ as per Observation~\ref{obs:treejoin-edges}.
Edges of type (iii) were not affected by the tree joining operation, therefore those edges are also either in $T_1^*$ or in $(r\rightarrow T_{2})^*$. Observe that $T_1$ and $(r\rightarrow T_2)$ are both quasi-subtrees of $T$. Therefore, by  lemma~\ref{lem:dual-edges-notroot} applied twice (once with $A=T_1$ and then with $A=T_2$), edges of type (iii) are in $T^*$.

At this point, to prove the edge inclusion, what remain to be shown is that edges of type (i) and (ii) of $T'$ are also in $T^*$.

Consider edges of type (ii), i.e. all the children $w_n,\ldots,w_1$ of $r$ in $(r\rightarrow T_2)^*$, from right to left.
We will show by induction that they are exactly the children of $r$ in $T^*$ also from right to left. For the base case, the rightmost child of $r$ in $(r\rightarrow T_2)^*$ is the root of $T_2$, same as in $T^*$ by Rule 1b. The induction step is as follows. If $w_i = ils_{(r\rightarrow T_2)^*}(w_{i-1})$, then $w_i=rmc_{(r\rightarrow T_2)}(w_{i-1})$ by Rule 3, as $((r\rightarrow T_2)^*)^*=(r\rightarrow T_2)$. Since $T_2$ is a subtree of $T$, $w_i = rmc_{T}(w_{i-1})$ and thus $w_i = ils_{T^*}(w_{i-1})$ by Rule 2, which completes the induction.

Finally, for the edges of type (i), consider all the children $v_m, \ldots, v_1$ of $r$ in $T_1^*$, from right to left. We will show that they are children of $r$ in $T^*$, also from right to left. 
For this, we set up an induction again. The base case examines $y_1$, which is the right-most child of $r$ in $T_1^*$, and remains so in $T^*$ by Rule 1b. For the inductive step, assume that $y_i = ils_{T_1^*}(y_{i-1})$, then again a similar reasoning as in the paragraph before yields that $y_i$ equals to $rmc_{T_1}(y_{i-1})$ (using Rule 3), also to $rmc_{T}(y_{i-1})$ (using that $T_1 = T \setminus T_2$), and finally to $ils_{T^*}(y_{i-1})$ (using Rule 2) which proves the induction. This concludes the proof, as all edges in $T'$ are therefore in $T^*$.
\end{proof}

  \noindent \textbf{Proof of Lemma~\ref{lem:dual-decomposition}.}

We prove this by induction over $n$.
The case $n=1$ is immediate. 
Assume that the lemma is true for $T'$ consisting of $r$ and $n-1$ subtrees $A_1,\ldots,A_{n-1}$. We now add $A_n$ as the rightmost child of the root of $T'$ in order to obtain $T$. 
Observe that setting $T_2=A_n$ and $T_1=T'$ satisfy the conditions of Lemma~\ref{lem:two-trees-join-is-dual}, therefore $T^*$ is equal to $(T')^* \curvearrowright (r\rightarrow A_n)^*$.
By induction,  $T^{\prime*}  =  (r\rightarrow A_1)^* \curvearrowright   (r\rightarrow A_2)^* \curvearrowright  \ldots \curvearrowright (r\rightarrow A_n)^*$, which concludes the proof.
\qed



\end{document}